\documentclass[11pt]{article}

\usepackage[utf8]{inputenc}
\usepackage[T1]{fontenc}
\usepackage[a4paper,margin=3cm]{geometry}
\usepackage{amsmath,amssymb,amsfonts,mathtools,amsthm}
\usepackage{mathrsfs}
\usepackage[numbers,sort&compress]{natbib}
\usepackage[colorlinks=true,linkcolor=blue,citecolor=blue,urlcolor=blue]{hyperref}

\numberwithin{equation}{section}

\newtheorem{theorem}{Theorem}[section]

\theoremstyle{definition}
\newtheorem{definition}[theorem]{Definition}
\theoremstyle{remark}
\newtheorem{remark}[theorem]{Remark}

\newcommand{\ac}[2]{\{#1,#2\}}

\newcommand{\cP}{\mathcal P}
\newcommand{\SWI}{\mathrm{SWI}}

\title{The 2D Smorodinsky--Winternitz II system and the Laguerre--Heun algebra}

\author{
Vutha Vichea Chea\textsuperscript{1,2}\thanks{E-mail: vutha.vichea.chea@umontreal.ca}
\quad
Luc Vinet\textsuperscript{1,2}\thanks{E-mail: luc.vinet@umontreal.ca}
\quad
Alexei Zhedanov\textsuperscript{3}\thanks{E-mail: zhedanov@yahoo.com}
\\[0.3cm]
\textsuperscript{1}\small Centre de recherches mathématiques, Université de Montréal, Montréal, Québec, Canada\\
\textsuperscript{2}\small Département de physique, Université de Montréal, Montréal, Québec, Canada\\
\textsuperscript{3}\small Euler International Mathematical Institute, Saint Petersburg, Russia
}

\date{}

\begin{document}

\maketitle

\begin{abstract}
We identify the quadratic symmetry algebra of the two-dimensional
Smorodinsky--Winternitz II system with a Laguerre-type confluent Heun algebra.
The system is separable in Cartesian and parabolic coordinates.  The
complementary Cartesian separation operator
\[
  Y=\partial_y^2-\omega^2y^2+\frac{1/4-c^2}{y^2}
\]
is of Laguerre type, while the parabolic integral \(W=L_2\) is its algebraic
Heun partner.  With \(Z=[Y,W]\), the defining relations are
\[
 [Y,Z]=16\omega^2W-2bY,\qquad
 [W,Z]=6Y^2-4HY+2bW+8\omega^2(1-c^2),
\]
where \(H\) is central.  This gives a direct superintegrable realization of
the Laguerre--Heun algebra.
\end{abstract}

\section{Introduction}

The two-dimensional system, referred to as the Smorodinsky--Winternitz (SW) II one
\cite{winternitz1966symmetry}, is described by the quantum Hamiltonian
\begin{equation}
  H=\partial_x^2+\partial_y^2
    -\omega^2(4x^2+y^2)+bx+\frac{1/4-c^2}{y^2}.
  \label{eq:intro-H}
\end{equation}
It is superintegrable and separates in Cartesian and parabolic coordinates.
We shall use the following two algebraically independent second-order
symmetries.  The Cartesian separation operator 
\begin{equation}
  L_1=\partial_x^2-4\omega^2x^2+bx,
  \label{eq:intro-L1}
\end{equation}
and the parabolic separation operator 
\begin{equation}
  L_2=\frac12\ac{M}{\partial_y}
      -y^2\left(\frac b4-x\omega^2\right)
      +\left(\frac14-c^2\right)\frac{x}{y^2},
  \qquad
  M=x\partial_y-y\partial_x .
  \label{eq:intro-L2}
\end{equation}
Both \(L_1\) and \(L_2\) commute with \(H\).  The complementary Cartesian
separation operator is
\begin{equation}
  Y=
  \partial_y^2-\omega^2y^2+\frac{1/4-c^2}{y^2}.
  \label{eq:intro-Y}
\end{equation}
Thus \(H=L_1+Y\).  Although this model has been studied extensively
\cite{miller2013classical}, its quadratic symmetry algebra does not seem to
have been explicitly identified with one of the standardized quadratic algebras like the Hahn and Racah algebras  that have been studied in the last decades and seen to occur in superintegrability as well as in many other areas.  The purpose of this paper is to show that the algebra
generated by the second-order constants of motion of the SW II model is a Laguerre-type
confluent Heun algebra.

The algebraic connection between Heun operators and tridiagonalization was
developed in \cite{GVZ2017}.  The basic idea is the following.  Given a
hypergeometric-type operator \(Y\) with a polynomial eigenbasis, its algebraic
Heun partner is the most general second-order operator \(W\) that maps
polynomials of degree \(n\) to polynomials of degree at most \(n+1\), or,
equivalently, acts tridiagonally in the eigenbasis of \(Y\).  For the Jacobi
operator this construction yields the standard Heun operator; for the
Laguerre operator it yields a confluent Heun operator.  The algebra generated
by \(Y\), \(W\), and \(Z=[Y,W]\) is a quadratic algebra, called here the
Laguerre--Heun algebra.

The Smorodinsky--Winternitz II system gives a natural physical realization of
this construction.  Cartesian separation diagonalizes the singular oscillator
operator
\[
  Y=\partial_y^2-\omega^2y^2+\frac{1/4-c^2}{y^2},
\]
which is of Laguerre type and such that \(H=L_1+Y\).  Parabolic separation
diagonalizes the second symmetry \(L_2\).  The key observation is that \(L_2\) is precisely the
algebraic Heun partner of \(Y\).  The resulting quadratic algebra is
\begin{equation}
 [Y,Z]=16\omega^2W-2bY,
 \qquad
 [W,Z]=6Y^2-4HY+2bW+8\omega^2(1-c^2),
 \label{eq:intro-main}
\end{equation}
with
\[
        W=L_2,\qquad Z=[Y,W],
\]
and with \(H\) central.

This identification also clarifies the nature of the Cartesian--parabolic
connection problem.  Since the relevant separated operator is of Laguerre
type, the parabolic integral should be viewed as a confluent Heun operator,
rather than as an operator associated with a finite Hahn-type recoupling
problem.

The commutation relations for the Smorodinsky--Winternitz II symmetries were
given in \cite{LetourneauVinet1995}.  We use those relations as the starting
point and rewrite them in the natural Laguerre--Heun generators.  We also
recall how the parabolic separated equations arise from
diagonalizing the second conserved quantity.  The tridiagonal action of the
parabolic symmetry in the Cartesian separated basis is then derived from the
quadratic algebra.  Finally, we contrast the result with the
Smorodinsky--Winternitz I model, whose Cartesian--polar separation problem is
governed by a Hahn-type algebra and dual Hahn overlap coefficients.

\section{Algebraic Heun operators and the Laguerre case}

We review the algebraic Heun construction in the form needed below.  Let
\(Y\) be a second-order hypergeometric-type operator with a basis of
eigenfunctions \(\{p_n\}_{n\geq0}\).  An algebraic Heun operator associated
with \(Y\) is an operator \(W\) such that
\begin{equation}
  Wp_n=\xi_{n+1}p_{n+1}+\eta_n p_n+\zeta_n p_{n-1}.
  \label{eq:tri-action}
\end{equation}
In other words, \(W\) is tridiagonal in the \(Y\)-eigenbasis.  Equivalently,
in the polynomial realization, \(W\) is the most general second-order
operator that maps polynomials of degree \(n\) into polynomials of degree at
most \(n+1\).

For the Jacobi operator this tridiagonalization gives the ordinary Heun
operator.  The algebra generated by the hypergeometric operator \(Y\), its
Heun partner \(W\), and their commutator \(Z=[Y,W]\), is a quadratic algebra
which extends the Racah algebra by Heun terms \cite{GVZ2017}.  Confluent
limits give the Laguerre and Hermite cases.  In the Laguerre case, the
ordinary Heun equation degenerates to a confluent Heun equation, and the
corresponding quadratic algebra takes a contracted form.

For our purposes, the Laguerre--Heun algebra may be characterized as follows.

\begin{definition}
A Laguerre--Heun algebra is an associative algebra generated by \(Y,W,Z\),
with \(Z=[Y,W]\), and with a central element \(H\), such that
\begin{align}
 [Y,Z]&=a_1W+a_2Y+a_3,                                      \label{eq:LH1}\\
 [W,Z]&=b_1Y^2+b_2Y+b_3W+b_4,                               \label{eq:LH2}
\end{align}
where \(a_i,b_i\) are central parameters.
\end{definition}

The first relation is linear in \(Y,W\), while the second contains a single
quadratic term.  This is the algebraic signature of the Laguerre, or
confluent, degeneration.  Such Lie-type Heun algebras and their
realizations, including those associated with \(\mathfrak{su}(1,1)\), were
also described in \cite{CrampeVinetZhedanov2020}.  In the
\(\mathfrak{su}(1,1)\) family, the Laguerre realization leads to a confluent
Heun operator.

\section{The Smorodinsky--Winternitz II symmetries}

We return to the operators introduced in the introduction in order to fix
notation and conventions.  We use the normalization
\begin{equation}
  H=\partial_x^2+\partial_y^2
    -\omega^2(4x^2+y^2)+bx+\frac{1/4-c^2}{y^2}.
  \label{eq:H}
\end{equation}
The signs correspond to the convention in which the kinetic energy is
\(\partial_x^2+\partial_y^2\).  Multiplication of \(H\) by \(-1\) gives the
usual Schrödinger sign convention.

The Cartesian second-order integral is
\begin{equation}
  L_1=\partial_x^2-4\omega^2x^2+bx.
  \label{eq:L1}
\end{equation}
The complementary Cartesian separation operator is
\begin{equation}
  Y=
  \partial_y^2-\omega^2y^2+\frac{1/4-c^2}{y^2}.
  \label{eq:Y}
\end{equation}
Thus \(H=L_1+Y\).  It is the operator $Y$ that belongs to the Laguerre family.

The parabolic second-order integral is
\begin{equation}
  L_2=\frac12\ac{M}{\partial_y}
      -y^2\left(\frac b4-x\omega^2\right)
      +\left(\frac14-c^2\right)\frac{x}{y^2},
  \qquad
  M=x\partial_y-y\partial_x .
  \label{eq:L2}
\end{equation}
Both \(L_1\) and \(L_2\) commute with \(H\).  Let
\begin{equation}
  R=[L_1,L_2].
\end{equation}
The quadratic symmetry algebra is
\begin{align}
 [L_1,R]&=16\omega^2L_2+2bL_1-2bH,                         \label{eq:LV1}\\
 [L_2,R]&=-6L_1^2+8HL_1-2bL_2-2H^2-8\omega^2(1-c^2).
                                                                    \label{eq:LV2}
\end{align}
These relations are those of the Smorodinsky--Winternitz II polynomial
algebra in the normalization of \cite{LetourneauVinet1995}.  The sign of
the \(2bL_2\) term in \eqref{eq:LV2} is fixed by the Jacobi identity once
the conventions \(R=[L_1,L_2]\) and \eqref{eq:LV1} have been chosen.

We shall not need the cubic Casimir relation for the identification with the
Laguerre--Heun algebra; it is therefore not written here.

\section{Cartesian and parabolic separation}

We now recall how the two separated coordinate systems are attached to
\(L_1\) and \(L_2\).

\subsection{Cartesian separation}

Let
\[
        \Psi(x,y)=X(x)Y_0(y)
\]
be a separated eigenfunction satisfying
\[
        H\Psi=E\Psi.
\]
Since \(L_1\) involves only the variable \(x\), Cartesian separation is the
simultaneous diagonalization of \(H\) and \(L_1\):
\begin{equation}
        L_1X=\lambda X.
\end{equation}
The complementary equation is governed by the operator
\begin{equation}
        Y=\partial_y^2-\omega^2y^2+\frac{1/4-c^2}{y^2}.
        \label{eq:Ydef}
\end{equation}
After the standard gauge and variable changes,
this is the Laguerre singular-oscillator problem.  This is the operator that
will play the role of the Laguerre operator in the algebraic Heun
construction.

\subsection{Parabolic separation}

Introduce parabolic coordinates \(u,v\) by
\begin{equation}
        x=\frac{u^2-v^2}{2},\qquad y=uv.
        \label{eq:parabolic}
\end{equation}
Then
\begin{equation}
        \partial_x^2+\partial_y^2
        =
        \frac{1}{u^2+v^2}\left(\partial_u^2+\partial_v^2\right).
        \label{eq:lap-par}
\end{equation}
Multiplying the eigenvalue equation \(H\Psi=E\Psi\) by \(u^2+v^2\), and
writing \(\Psi(u,v)=U(u)V(v)\), gives
\begin{align}
0={}&
\left[
\partial_u^2
-\omega^2u^6+\frac b2u^4-Eu^2+\frac{1/4-c^2}{u^2}
\right]U(u)\,V(v)                                                   \notag\\
&+
U(u)\left[
\partial_v^2
-\omega^2v^6-\frac b2v^4-Ev^2+\frac{1/4-c^2}{v^2}
\right]V(v).
\label{eq:par-sep}
\end{align}
Thus parabolic separation is achieved by imposing
\begin{align}
\left[
\partial_u^2
-\omega^2u^6+\frac b2u^4-Eu^2+\frac{1/4-c^2}{u^2}
\right]U(u)&=\mu\,U(u),                                      \label{eq:u-eq}\\
\left[
\partial_v^2
-\omega^2v^6-\frac b2v^4-Ev^2+\frac{1/4-c^2}{v^2}
\right]V(v)&=-\mu\,V(v).                                    \label{eq:v-eq}
\end{align}
The separation constant \(\mu\) is the eigenvalue of a second-order symmetry
operator.  Indeed, define
\begin{align}
\cP=\frac{1}{u^2+v^2}
\Bigg\{&
v^2\left(
\partial_u^2-\omega^2u^6+\frac b2u^4+\frac{1/4-c^2}{u^2}
\right)                                                       \notag\\
&-
u^2\left(
\partial_v^2-\omega^2v^6-\frac b2v^4+\frac{1/4-c^2}{v^2}
\right)
\Bigg\}.
\label{eq:P}
\end{align}
On a solution of \(H\Psi=E\Psi\), the terms involving \(E\) cancel and
\[
        \cP\Psi=\mu\Psi.
\]
A direct calculation using \eqref{eq:parabolic} gives
\begin{equation}
        \cP=-2L_2.
        \label{eq:P-L2}
\end{equation}
Thus diagonalizing the parabolic integral \(L_2\) is exactly the operator
form of parabolic separation.

\begin{remark}
Equations \eqref{eq:u-eq} and \eqref{eq:v-eq} are not hypergeometric
equations in general.  This is one reason why the Cartesian--parabolic
connection problem should be viewed as a confluent Heun spectral problem
rather than as a standard Askey-scheme polynomial overlap problem.
\end{remark}

\section{The Laguerre--Heun presentation of the symmetry algebra}

We now pass from the generators \(L_1,L_2\) to the Laguerre--Heun generators
\begin{equation}
        Y=\partial_y^2-\omega^2y^2+\frac{1/4-c^2}{y^2},
        \qquad W=L_2,\qquad Z=[Y,W].
        \label{eq:assignment}
\end{equation}
Since \(H=L_1+Y\), one has \(L_1=H-Y\).  As \(H\) is central,
\[
        Z=[H-L_1,L_2]=-[L_1,L_2]=-R.
\]
Using \eqref{eq:LV1}--\eqref{eq:LV2}, we obtain the following result.

\begin{theorem}
The Smorodinsky--Winternitz II symmetry algebra is generated by \(Y,W,Z\),
with \(H\) central and \(Z=[Y,W]\), and satisfies
\begin{align}
 [Y,Z]&=16\omega^2W-2bY,                                      \label{eq:main1}\\
 [W,Z]&=6Y^2-4HY+2bW+8\omega^2(1-c^2).                         \label{eq:main2}
\end{align}
Equivalently,
\begin{align}
 [Y,[Y,W]]&=16\omega^2W-2bY,                                  \label{eq:double1}\\
 [W,[W,Y]]&=-6Y^2+4HY-2bW-8\omega^2(1-c^2).                    \label{eq:double2}
\end{align}
These are the defining relations of the Laguerre--Heun algebra associated
with the Smorodinsky--Winternitz II system.
\end{theorem}

\begin{proof}
Substitute \(L_1=H-Y\), \(L_2=W\), and \(R=-Z\) into
\eqref{eq:LV1}.  Since
\[
        [L_1,R]=[H-Y,-Z]=[Y,Z],
\]
the first relation gives
\[
 [Y,Z]
 =
 16\omega^2W+2b(H-Y)-2bH
 =
 16\omega^2W-2bY.
\]
Similarly, from \eqref{eq:LV2},
\[
        [L_2,R]=[W,-Z]=-[W,Z].
\]
Hence
\begin{align*}
[W,Z]
&=
6(H-Y)^2-8H(H-Y)-2bW+2H^2+8\omega^2(1-c^2)  \\
&=
6Y^2-4HY+2bW+8\omega^2(1-c^2).
\end{align*}
This proves \eqref{eq:main1}--\eqref{eq:main2}.  The double-commutator form
follows from \(Z=[Y,W]\).
\end{proof}

In the notation of the defining relations \eqref{eq:LH1}--\eqref{eq:LH2}, this realization
corresponds to the central structure constants
\[
        a_1=16\omega^2,\qquad a_2=-2b,\qquad a_3=0,
\]
and
\[
        b_1=6,\qquad b_2=-4H,\qquad b_3=2b,\qquad
        b_4=8\omega^2(1-c^2).
\]

The theorem gives the precise meaning of the statement that the parabolic
integral is the Heun partner of the Cartesian Laguerre operator.  The relation
\[
        [Y,[Y,W]]=16\omega^2W-2bY
\]
implies that, in a basis diagonalizing \(Y\), the operator \(W\) is
tridiagonal.  Thus the algebraic fact that \(W\) is a Heun operator is the
same as the separability fact that \(W=L_2\) is the operator whose
diagonalization gives parabolic coordinates.

\section{The representation in the Cartesian basis}

We now describe the action of the parabolic integral \(L_2\) in a basis
diagonalizing the Laguerre operator
\[
        Y=\partial_y^2-\omega^2y^2+\frac{1/4-c^2}{y^2}.
\]
This gives the concrete representation-theoretic meaning of the statement
that \(L_2\) is the algebraic Heun partner of \(Y\).

Fix an energy eigenspace, so that \(H=E\) is scalar.  The operator \(Y\) is
the one-dimensional singular oscillator.  Its polynomial sector realizes a
positive discrete series representation of \(\mathfrak{su}(1,1)\).  In the
normalization used here, the spectrum of \(Y\) on this basis is equally
spaced:
\begin{equation}
        Y e_n=\lambda_n e_n,
        \qquad
        \lambda_{n+1}-\lambda_n=4\omega .
        \label{eq:Y-spectrum}
\end{equation}
Equivalently,
\begin{equation}
        \lambda_n=\lambda_0+4\omega n.
        \label{eq:lambda-linear}
\end{equation}
Here \(\lambda_0\) is determined by the Bargmann index of the
\(\mathfrak{su}(1,1)\) representation, hence by the parameter \(c\), together
with the sign convention chosen for the differential operator.  Only the
overall sign and ordering of the sequence depend on conventions; the constant
spacing is fixed by the \(\mathfrak{su}(1,1)\) discrete-series realization of
the singular oscillator.

Set \(W=L_2\).  From
\[
        [Y,[Y,W]]=16\omega^2W-2bY,
\]
one immediately obtains, for \(m\neq n\),
\begin{equation}
        \big((\lambda_m-\lambda_n)^2-16\omega^2\big)
        \langle e_m,W e_n\rangle=0.
        \label{eq:tri-condition}
\end{equation}
Since adjacent eigenvalues of \(Y\) differ by \(4\omega\), the only
off-diagonal matrix elements of \(W\) can occur between nearest-neighbour
eigenspaces.  Thus \(W\) is tridiagonal in the \(Y\)-eigenbasis.  We may
therefore write
\begin{equation}
        W e_n
        =
        A_n e_{n+1}+B_n e_n+C_n e_{n-1},
        \qquad C_0=0.
        \label{eq:W-action}
\end{equation}

The diagonal coefficient follows from the diagonal part of
\([Y,[Y,W]]=16\omega^2W-2bY\).  Since the left-hand side has zero diagonal
part, one finds
\[
        0=16\omega^2 B_n-2b\lambda_n,
\]
and hence
\begin{equation}
        B_n=\frac{b}{8\omega^2}\lambda_n.
        \label{eq:Bn}
\end{equation}

The products of the off-diagonal coefficients are determined by the second
Laguerre--Heun relation.  Let
\begin{equation}
        U_n=A_{n-1}C_n,\qquad U_0=0.
        \label{eq:Un-def}
\end{equation}
Using \(Z=[Y,W]\) and \eqref{eq:Y-spectrum}, one has
\begin{equation}
        Z e_n=4\omega A_n e_{n+1}-4\omega C_n e_{n-1}.
        \label{eq:Z-action}
\end{equation}
A direct computation gives
\begin{equation}
        [W,Z]e_n
        =
        8\omega\,(U_{n+1}-U_n)e_n
        +2bA_n e_{n+1}
        +2bC_n e_{n-1}.
        \label{eq:WZ-action}
\end{equation}
On the other hand, the quadratic relation
\[
        [W,Z]=6Y^2-4EY+2bW+8\omega^2(1-c^2)
\]
gives, on \(e_n\),
\begin{align}
        [W,Z]e_n
        &=
        \left(6\lambda_n^2-4E\lambda_n
        +2bB_n+8\omega^2(1-c^2)\right)e_n             \notag\\
        &\quad
        +2bA_n e_{n+1}
        +2bC_n e_{n-1}.
        \label{eq:WZ-action-2}
\end{align}
The off-diagonal parts agree identically, while the diagonal part yields
\begin{equation}
        U_{n+1}-U_n
        =
        \frac{
        6\lambda_n^2-4E\lambda_n
        +\dfrac{b^2}{4\omega^2}\lambda_n
        +8\omega^2(1-c^2)}
        {8\omega}.
        \label{eq:Un-rec}
\end{equation}
Thus
\begin{equation}
        U_n
        =
        \frac{1}{8\omega}
        \sum_{k=0}^{n-1}
        \left(
        6\lambda_k^2
        -4E\lambda_k
        +\frac{b^2}{4\omega^2}\lambda_k
        +8\omega^2(1-c^2)
        \right).
        \label{eq:Un-sum}
\end{equation}
If \(\lambda_n=\lambda_0+4\omega n\), the sum can be evaluated explicitly:
\begin{align}
U_n
={}&
\frac{1}{8\omega}
\Bigg[
n\left(6\lambda_0^2+
\left(-4E+\frac{b^2}{4\omega^2}\right)\lambda_0
+8\omega^2(1-c^2)\right)                         \notag\\
&\quad
+24\omega\lambda_0\,n(n-1)
+2\omega\left(-4E+\frac{b^2}{4\omega^2}\right)n(n-1)       \notag\\
&\quad
+16\omega^2 n(n-1)(2n-1)
\Bigg].
\label{eq:Un-explicit}
\end{align}
In an orthonormal realization in which \(W\) is symmetric, one may choose
\begin{equation}
        A_n=C_{n+1}=\sqrt{U_{n+1}},
        \label{eq:symmetric-choice}
\end{equation}
up to harmless phase conventions.  Therefore
\begin{equation}
        L_2 e_n
        =
        \sqrt{U_{n+1}}\, e_{n+1}
        +
        \frac{b}{8\omega^2}\lambda_n e_n
        +
        \sqrt{U_n}\, e_{n-1}.
        \label{eq:L2-action}
\end{equation}
This is the desired tridiagonal action of the parabolic symmetry in the
Cartesian separated basis.

For a finite-dimensional bound-state representation, say
\[
        n=0,1,\ldots,N,
\]
one has the boundary conditions
\begin{equation}
        U_0=0,\qquad U_{N+1}=0.
        \label{eq:finite-boundary}
\end{equation}
In such a finite-dimensional realization, the second condition is the usual
truncation condition and determines the admissible central character,
equivalently the allowed energy \(E\), in terms of the representation label
\(N\) and the parameters.  In this form, the parabolic separation problem
becomes the spectral problem of the tridiagonal matrix \eqref{eq:L2-action}.
Since the coefficients \(U_n\) are cubic in \(n\) in general, the resulting
connection problem is naturally of confluent Heun type rather than a classical
finite dual-Hahn problem.

\begin{remark}
The spacing in \eqref{eq:Y-spectrum} is the singular-oscillator
\(\mathfrak{su}(1,1)\) spacing in the present normalization.  If one rescales
the Hamiltonian or the symmetry operators, both this spacing and the
coefficient \(16\omega^2\) in the double commutator are rescaled
accordingly; the tridiagonal structure is unchanged.
\end{remark}

\section{Comparison with the Smorodinsky--Winternitz I system}

It is useful to contrast the preceding result with the 
Smorodinsky--Winternitz I system, namely the isotropic singular oscillator
\begin{equation}
H_{\SWI}
=
\partial_x^2+\partial_y^2
-\omega^2(x^2+y^2)
+\frac{\alpha}{x^2}+\frac{\beta}{y^2}.
\label{eq:SWI}
\end{equation}
This system separates in Cartesian and polar coordinates.  Its quadratic
symmetry algebra is of Hahn type, or equivalently a degeneration of the Racah
algebra.  In representation-theoretic terms, the change of basis between
Cartesian and polar separated eigenfunctions is a finite recoupling problem;
the associated overlap coefficients are described, in the standard finite
models, by dual Hahn polynomials, up to conventional changes of parameters
and normalization \cite{Post2011,KoekoekLeskySwarttouw2010}.

The situation for the Smorodinsky--Winternitz II system is different.  The
Cartesian separated operator singled out in this paper is the Laguerre-type
operator
\[
        Y=H-L_1,
\]
and the parabolic integral
\[
        W=L_2
\]
is its confluent Heun partner.  Thus the Cartesian--parabolic connection
problem is not expected, in general, to reduce to a classical finite
dual-Hahn overlap problem.  It is instead governed by the spectral theory of
the corresponding Laguerre--Heun operator.

This contrast may be summarized schematically as follows:
\begin{center}
\renewcommand{\arraystretch}{1.25}
\begin{tabular}{c|c|c|c}
system & coordinates & symmetry algebra & connection problem \\
\hline
SW I & Cartesian/polar & Hahn type &
dual Hahn-type finite recoupling \\
SW II & Cartesian/parabolic & Laguerre--Heun &
confluent Heun spectral problem
\end{tabular}
\end{center}
The second line is the content of the present paper.  This comparison should
not be read as a claim that the two systems are parallel in every respect.
Rather, it emphasizes the algebraic distinction between a finite Hahn
recoupling problem and a confluent Heun connection problem.

\section{Relation with the finite Heun--Hahn algebra}

The finite-grid Heun--Hahn algebra arises from the tridiagonalization of the
Hahn operator on a uniform lattice \cite{VZ2019}.  It is an important member
of the family of Heun algebras, but it is not the generic algebra of the
Smorodinsky--Winternitz II system.

The reason is structural.  The finite Heun--Hahn algebra is tied to a finite
difference Hahn operator and hence to a finite discrete spectral grid.  The
operator \(Y\) in the present paper is instead a Laguerre-type differential
operator, obtained after confluence.  Its Heun partner is
therefore a confluent Heun operator.  The resulting algebra has the
Laguerre--Heun form \eqref{eq:main1}--\eqref{eq:main2}, not the finite
Heun--Hahn form.

There are parameter specializations in which the relations simplify further
and resemble finite Hahn-type formulas.  Such coincidences should not obscure
the main point: the natural algebraic interpretation of the full
Smorodinsky--Winternitz II symmetry algebra is the Laguerre--Heun one.

\section{Conclusion}

In summary, the symmetry algebra of the Smorodinsky--Winternitz II system is a Laguerre--Heun algebra or put differently, this superintegrable model provides a natural
realization of this quadratic algebra.  The central observation is that one the two complementary Cartesian separation operators is the Laguerre one:

\[
        Y=
        \partial_y^2-\omega^2y^2+\frac{1/4-c^2}{y^2}.
\]
The parabolic integral
\[
        W=L_2
\]
is the algebraic Heun partner of \(Y\).  The commutation relations are
\[
 [Y,Z]=16\omega^2W-2bY,
 \qquad
 [W,Z]=6Y^2-4HY+2bW+8\omega^2(1-c^2),
\]
where \(Z=[Y,W]\) and \(H\) is central.  The tridiagonal action
\[
        L_2 e_n
        =
        \sqrt{U_{n+1}}\, e_{n+1}
        +
        \frac{b}{8\omega^2}\lambda_n e_n
        +
        \sqrt{U_n}\, e_{n-1}
\]
makes explicit the representation-theoretic meaning of this identification.

This formulation makes transparent that the associated connection or overlap problem is
governed by the representations of this confluent Heun algebra.

The comparison with SW I is instructive.  In SW I, the Cartesian--polar
separation problem leads to a Hahn-type algebra and to dual Hahn overlap
coefficients.  For SW II, the parabolic integral is instead the confluent
Heun partner of a Laguerre operator.  This is the essential algebraic
distinction between the two Smorodinsky--Winternitz systems. Furthermore, in the wake of the recent work on the spectrum generating algebra of the generic superintegrable model on the two-sphere \cite{Crampe2026dyn}, it would be of interest to determine the (rank two) dynamical algebras of these two Smorodinsky-Winternitz models.

\section*{Acknowledgments}

VVC is enjoying a scholarship from the Fonds de Recherche du Québec - Nature et Technologie (FRQ-NT). LV is funded in part through a discovery grant of the Natural Sciences and Engineering Research Council (NSERC) of Canada. 
 AZ is supported by the Ministry of Science and Higher Education of the Russian Federation (agreement no. 075--15--2025--343).


\begin{thebibliography}{99}

\bibitem{winternitz1966symmetry}
P.~Winternitz, Ya.~A. Smorodinsky, M.~Uhlir and I.~Fri\v{s},
Symmetry groups in classical and quantum mechanics,
\emph{Yad. Fiz.} \textbf{4} (1966), 625--635;
English translation:
\emph{Sov. J. Nucl. Phys.} \textbf{4} (1967), 444--450.

\bibitem{miller2013classical}
W.~Miller Jr., S.~Post and P.~Winternitz,
Classical and quantum superintegrability with applications,
\emph{J. Phys. A: Math. Theor.} \textbf{46} (2013), 423001.

\bibitem{LetourneauVinet1995}
P.~Létourneau and L.~Vinet,
Superintegrable systems: polynomial algebras and quasi-exactly solvable
Hamiltonians,
\emph{Ann. Phys.} \textbf{243} (1995), 144--168.

\bibitem{GVZ2017}
F.~A. Grünbaum, L.~Vinet and A.~Zhedanov,
Tridiagonalization and the Heun equation,
\emph{J. Math. Phys.} \textbf{58} (2017), 031703,
\texttt{arXiv:1602.04840}.

\bibitem{CrampeVinetZhedanov2020}
N.~Crampé, L.~Vinet and A.~Zhedanov,
Heun algebras of Lie type,
\emph{Proc. Amer. Math. Soc.} \textbf{148} (2020), 1079--1094,
\texttt{arXiv:1904.10643}.

\bibitem{VZ2019}
L.~Vinet and A.~Zhedanov,
The Heun operator of Hahn-type,
\emph{Proc. Amer. Math. Soc.} \textbf{147} (2019), 2987--2998,
\texttt{arXiv:1808.00153}.

\bibitem{Post2011}
S.~Post,
Models of quadratic algebras generated by superintegrable systems in 2D,
\emph{SIGMA} \textbf{7} (2011), 036, 20 pp.,
\texttt{arXiv:1104.0734}.

\bibitem{KoekoekLeskySwarttouw2010}
R.~Koekoek, P.~A. Lesky and R.~F. Swarttouw,
\emph{Hypergeometric Orthogonal Polynomials and Their \(q\)-Analogues},
Springer Monographs in Mathematics, Springer, Berlin, 2010.

\bibitem{Crampe2026dyn}
N.~Crampé, Q.~Labriet, L.~Morey, S.~Tsujimoto, L.~Vinet, A.~Zhedanov,
\emph{The dynamical algebra of the generic superintegrable model on the two-sphere}
\texttt{arXiv:2604.26122}.
\end{thebibliography}
\end{document}